%% file: main.tex
\documentclass[submission,copyright,creativecommons]{eptcs}
\providecommand{\event}{LFMTP'23} 
\usepackage{iftex,float,tipa,verbatim,dsfont,prftree,amssymb,amsthm ,hyperref}  
\usepackage{tikz}
\usepackage{tikz-cd}

\floatstyle{boxed} 
\restylefloat{figure}

\ifpdf
  \usepackage{underscore}         
  \usepackage[T1]{fontenc}        
\else
  \usepackage{breakurl}           
\fi

\input{macro.tex}

\title{An Interpretation of \EHAw{} inside \HAw{}}
\author{F\'elix Castro
\institute{IRIF\\
    Universit\'e Paris Cit\'e\\
    Paris, France}
\institute{IMERL\\
Facultad de Ingenier\'ia, Universidad de la
Rep\'ublica\\
Montevideo, Uruguay}
\email{castro@irif.fr}
}
\def\titlerunning{An interpretation of \EHAw{} inside \HAw{}}
\def\authorrunning{F. Castro}
\begin{document}
\maketitle

\begin{abstract}
Higher Type Arithmetic (\HAw{}) is a first-order 
many-sorted theory.
It is a conservative extension of Heyting Arithmetic obtained by extending the syntax of terms to all of System T: the objects of interest here are the functionals of ``higher types''.
While equality between natural numbers is specified by the axioms of 
Peano, how can equality between functionals be defined? From this 
question, different versions of \HAw{} arise, such as an extensional 
version (\EHAw{}) and an intentional version (\IHAw{}).
In this work, we will see how the study of partial equivalence relations 
leads us to design a translation by parametricity from \EHAw{} to \HAw{}. 
\end{abstract}

\section{Introduction}
\label{sec:Intro}

In second-order logic, it can be shown as a meta-theorem that
two extensionally equal predicates satisfy the same properties.
It is not the case in higher-order logic: this is due
to the potential existence of non-extensional (higher-order) predicates.
However, Gandy showed that axioms of extensionality could be consistently
added by restraining the range of quantification to \emph{extensional}
elements~\cite{gan56}.
A similar phenomenon occurs in Higher Type Arithmetic (\HAw{}): 
one cannot prove
in \HAw{} that two extensionally equal functions satisfy the same 
formulas. It can be seen for instance by working in the model
of Hereditary Recursive Operations~\HRO{}~\cite{tro73} where a
functional can inspect the source code of its argument. But, again,
axioms of extensionality can be added without loss of consistency: Zucker
showed that every model of \NHAw{} (Higher Type Arithmetic with
equality at all levels of sort)
can be turned into a model
of \EHAw{} (Higher Type Arithmetic with
extensional equality at all levels of sort)~\cite{zuc71}.

In this work we tackle a similar problem.
Starting from \HAw{}, we show
that an extensional equality can be consistently 
added at all levels of sorts. 
Taking inspiration from syntactical
models of type theory~\cite{bpt17,abkt19}, we choose to do it in a syntactical
fashion: we design an interpretation of \EHAw{} in \HAw{} that
we express as a translation between two proof systems (without reduction
rules). Concretely, we will compile a language with extensional
equality at all levels of sorts to a language that merely has equality
in the sort $\Nat{}$. It will be done using techniques
of parametricity, as one goal of this paper is to
emphasize that parametricity can be used to
extend equality.

After exposing a proof system \LHAw{} that captures Higher Type 
Arithmetic (Section \ref{sec:TypeSystem}),
we will study families (indexed by
the sorts of System~\T{}) of (internal) partial
equivalence relations that could be used to extend equality
(Section \ref{sec:PEE}).
In particular, we will compare two potential candidates: 
\begin{enumerate}
    \item a family $\Ext_{\sigma}$ generated from equality
    (over $\Nat$) in an extensional fashion;
    \item a family $\Eqpm_{\sigma}$ generated from equality 
    (over $\Nat$) in a
    way reminiscent of binary parametricity~\cite{rey83}.
\end{enumerate}
While the former is reflexive, the latter is not. But being reflexive
is not desirable in this context. 
Indeed, as explained above, one 
needs to restrict the range of quantifications before extending
equality: specifically we will restrict quantifications on a sort
$\sigma$ to the domain of $\Eqpm_{\sigma}$. 
Our first translation will be used to show that each 
closed term of System~\T{} 
is indeed in the domain of $\Eqpm$: we translate
judgments of System~\T{} into judgments of \LHAw{}
and we follow
the typical translation by parametrecity, as it will allow us to show
that typed terms satisfy the relation linked to their type
\cite{rey83,wad89, BJP10}.
Finally, by keeping the idea of a translation by
parametricity,
we will translate a proof system
\LEHAw{} (capturing \EHAw{}) to \LHAw{} 
(Section~\ref{sec:ST}). Before concluding,
we will compare our result and our methodology
to related work~(Section~\ref{sec:RW}).

\section{A proof system for Higher Type Arithmetic}
\label{sec:TypeSystem}

\subsection{System T}
\label{subsec:SystemT}
We use a version of
G\"odel's System~\T{}
obtained by extending the simply typed 
$\lambda$-calculus (\`a la Church) with a type constant $\Nat$ and native
constructors
to use it.
Terms, sorts and signatures of System~\T{} are described as follows:
$$\hspace*{0.5cm} 
\begin{array}{l@{\hspace*{1cm}}rcl}
\textrm{Sorts} & \sigma,\tau &::= & 
    \Nat \ \vert \ \sigma \rightarrow \tau \\[0.2cm]
\textrm{Terms} & t,u &::= & 
    x^{\sigma} \ \vert \ \lambda x^{\sigma}. t \ \vert \ t u \\
	&       &     & \vert \ 0 \ \vert \ \succe \, t \ \vert \
	                    \Rec^{\sigma} \, t \, u \, v \\[0.2cm]
\textrm{Signatures} & \Delta &::= & 
    \emptyset \ \vert \ \Delta , x^{\sigma} \\[0.2cm]
\end{array}$$
System~\T{} is presented in Church's style so terms come 
associated with a unique sort.	Nevertheless, we use
a type system (see Figure \ref{fig1} 
page \pageref{fig1}) to take into account in which
signature (or environment) a term is considered. We may omit
sort annotations on variables.

\begin{figure}[t]
		\centering
	\begin{displaymath}
	\begin{array}{cc}
    
    \multicolumn{2}{c}{

		\prfbyaxiom{}{
\Delta_1 ,x^{\sigma} , \Delta_2 \vdash_{\T} x^{\sigma}: \sigma} 
}
     \\[0.2cm]
 
	\prftree
    	{\prfassumption{ \Delta, x^{\sigma} \vdash_{\T} t: \tau}}
    	{\Delta \vdash_{\T}  \lambda x^{\sigma}. t: \sigma \rightarrow \tau
    	} & 
 
	\prftree
    	{\prfassumption
    	    {\Delta \vdash_{\T} t: \sigma \rightarrow \tau}}
    	{\prfassumption{ \Delta \vdash_{\T} u: \sigma}}
    	{\Delta \vdash_{\T} t u: \tau } \\[0.2cm] 
   
	\prfbyaxiom{}
    	{\Delta \vdash_{\T} 0: \mathbf{N} } &   

	\prftree
	    {\prfassumption{ \Delta \vdash_{\T} t: \Nat}}
		{\Delta \vdash_{\T} \succe \, t: \Nat }  \\[0.2cm]

	\multicolumn{2}{c}{
		\prftree
		{\prfassumption{ \Delta \vdash_{\T} t: \sigma}}
		{\prfassumption{ \Delta \vdash_{\T} u: 
		    \sigma \rightarrow \Nat \rightarrow \sigma}}
		{\prfassumption{ \Delta \vdash_{\T} v: \Nat}}
	   {\Delta \vdash_{\T} \Rec^{\sigma} \, t \, u \, v: \sigma}   
   }\\[0.2cm]
   
\end{array}
\end{displaymath}
\caption{\label{fig1} Derivation in System~\T{}}
\end{figure}

We consider the following rules on terms
    $$\begin{array}{rcl}
         (\lambda x . t) \, u & \succ & t [x:: = u]  \\
         \Rec \, t \, u \, 0 & \succ & t \\
         \Rec \, t \, u \, (\succe \, v) & \succ & 
             u \, (\Rec \, t \, u \, v) \, v
    \end{array}$$
from which we generate reduction and congruence
         $t  \leadsto  u$  and 
         $t  \cong  u$
as respectively
the least reflexive, transitive and closed by congruence relation 
containing
$\succ$ and 
the least closed by congruence equivalence relation containing 
$\succ$. We define a substitution $\theta$ to be a finite
function from variables to terms. The action of a substitution
$\theta$ on terms, denoted $t[\theta]$, corresponds to the 
simultaneous substitutions of free variables $x$ in the
domain of $\theta$ by~$\theta(x)$.

Metatheoretical results about System~\T{} can be found 
in the book of Girard, Lafont and Taylor~\cite{gir89}, for instance: 

\begin{enumerate}
    \item terms of System~\T{}
    are strongly normalizable;
    \item closed normal terms of type 
    $\Nat$ are of the form $\succe^n \, 0$, closed normal terms
    of type $\sigma \rightarrow \tau$ are of the form 
    $\lambda x^{\sigma}. t$.
    \end{enumerate}

Finally, we will use the two following facts.

\begin{fact}
A generalized version of the weakening rule is admissible for
this system: 
    \begin{center}
        \begin{tabular}{crcl}
             if & $\Delta \subseteq \Delta'$ and 
             $\Delta \vdash_{\T} t: \sigma$ & then &
             $\Delta' \vdash_{\T} t: \sigma$
        \end{tabular}
    \end{center}    
    where $\Delta \subseteq \Delta'$ is interpreted as the set-theoretic
    inclusion (while seeing signatures as sets).
\end{fact}

\begin{fact}
If $\theta$ is a substitution then
              $t \cong u $  implies 
             $t[\theta] \cong u[\theta]$.
\end{fact}

\subsection{Higher Type Arithmetic}
\label{subsec:HAw}

Higher Type Arithmetic (\HAw{}) is a theory of many-sorted first
order logic.
It is a conservative extension of $\mathbf{HA}$ obtained by 
extending the term language to the System~\T{}. 
Models of \HAw{} are described in the book of Troelstra~\cite{tro73}, in 
particular
the following will be used in the sequel:
\begin{enumerate}
    \item the set-theoretic model \M{} defined by
    $$\begin{array}{rcl}
             \M_{\Nat} & \equiv & \mathds{N}  \\
             \M_{\sigma \rightarrow \tau} & \equiv &
             \M_{\tau}^{\M_{\sigma}}
        \end{array}$$
        \item the model of Hereditary Recursive Operations~\HRO{} 
    defined by
    $$\begin{array}{rcl}
             \HRO_{\Nat} & \equiv & \mathds{N}  \\
             \HRO_{\sigma \rightarrow \tau} & \equiv &
             \{e \in \mathds{N} \ \vert \ \forall n \in 
             \HRO_{\sigma} \ \{e\}(n) \downarrow \in \HRO_{\tau} \}
        \end{array}$$
where $\{e\}(n) \downarrow \in E$ means that the computation 
of 
    the function of index $e$ terminates on the input $n$ and 
    that the result of this computation is in $E$.
\end{enumerate}
We define a proof system \LHAw{} that
captures \HAw{}. Formulas, proof terms and contexts of~
\LHAw{} are generated by the following grammar: 
$$\begin{array}{l@{\hspace*{1cm}}rcl}\textrm{Formulas} & \Phi,\Psi &::= & 
    t = u \ \vert \ \bot \ \vert \ \nullt (t) \\
 &          &      & 
    \vert \ \Phi \Rightarrow \Psi \ \vert \ \Phi \wedge \Psi \\ 
&          &      & 
    \vert \ \forall x^{\sigma} \Phi  \ \vert 
        \ \exists x^{\sigma} \Phi \\[0.2cm]
\textrm{Proof terms}& M,N &::=  & 
    \xi \ \vert \ \refle \, t \ \vert \ \peel^{t,u}(M,\hat{x}.\Phi, N) 
        \ \vert \ \efq(M,\Phi) \\
 &     &      & 
    \vert \ \lambda \xi. M \ \vert \ M \, N \\
&     &      & 
    \vert \ (M,N) \ \vert M.1 \ \vert \ M.2 \\
&     &      & 
    \vert \ \lambda x^{\sigma}.M \ \vert \ M \, t \\
&     &      & 
    \vert \ [t,M] \ \vert \ \letc \ [x , \xi ]:= M \ \inc \ N \\
&	  &      & 
    \vert \ \Ind (\hat{x}. \Phi , M , N ,t) \\[0.2cm]
\textrm{Contexts} & \Gamma &::= & 
    \emptyset \ \vert \ \Gamma , \xi: \Phi
	\end{array}$$

 \begin{figure}[t]
	\centering
	\begin{displaymath}
	\begin{array}{cc}

	\multicolumn{2}{c}{
		\begin{array}{ccc}
	\prftree[r]{\footnotesize$(\xi: \phi \in \Gamma)$}
	    {\prfassumption{(\Delta ; \Gamma) \ \wfp}}
	    {\Delta ; \Gamma \vdash \xi: \Phi} &
	
		\prftree[r]{\footnotesize$(\FV(\Phi) \subseteq \Delta)$}
			{\prfassumption{\Delta ; \Gamma \vdash M: \bot}}
			{\Delta ; \Gamma \vdash \efq (M,\Phi): \Phi}  &
					
		\prftree[r]{\footnotesize$(\Phi \simeq \Psi)$}
			{\prfassumption{\Delta ; \Gamma \vdash M: \Phi}}
			{\Delta ; \Gamma \vdash M: \Psi } 
		\end{array} 
	}
		\\[0.2cm]

	\prftree
		{\prfassumption
		{ \Delta ; \Gamma , \xi: \Phi \vdash M: \Psi}}
		{\Delta; \Gamma \vdash \lambda \xi.M: \Phi \Rightarrow \Psi } 
					& 
	\prftree
		{\prfassumption{ \Delta ; \Gamma \vdash M: \Phi \Rightarrow \Psi}}
		{\prfassumption{ \Delta ; \Gamma \vdash N: \Phi}}
		{\Delta ; \Gamma \vdash M \, N: \Psi } 
		\\[0.2cm] 
	
	\prftree
		{\prfassumption{ \Delta ; \Gamma \vdash M_1: \Phi_1}}
		{\prfassumption { \Delta ; \Gamma \vdash M_2: \Phi_2}}
		{\Delta ; \Gamma \vdash (M_1, M_2): \Phi_1 \wedge \Phi_2}
					&
	\prftree[r]{\footnotesize$(i=1,2)$}
		{\prfassumption {\Delta ; \Gamma \vdash M: \Phi_1 \wedge \Phi_2}}
		{\Delta ; \Gamma \vdash M.i: \Phi_i} 
		\\[0.2cm]
    
	\prftree[r]{\footnotesize$(x^{\sigma} \notin \FV(\Gamma))$}
		{\prfassumption{ \Delta, x^{\sigma} ; \Gamma \vdash M: \Phi}}
		{\Delta; \Gamma \vdash \lambda x^{\sigma}. M: \forall x^{\sigma}  \Phi } 
					&
	\prftree
		{\prfassumption{ \Delta ; \Gamma \vdash M:  \forall x^{\sigma} \Phi}}
		{\prfassumption { \Delta \vdash_{\T} t: \sigma}}
		{\Delta ; \Gamma \vdash M \, t: \Phi[x^{\sigma}:= t] } \\[0.2cm] 

	\prftree
		{\prfassumption{ \Delta; \Gamma \vdash M:  \Phi[x^{\sigma}:= t]}}
		{\prfassumption { \Delta \vdash_{\T} t: \sigma}}
		{\Delta ; \Gamma \vdash [t,M]:  \exists x^{\sigma}  \Phi}
					&
	\prftree[r]{\footnotesize$(x^{\sigma} \notin \FV(\Gamma,\Psi))$}
		{\prfassumption{ \Delta ; \Gamma \vdash M: \exists x^{\sigma} \Phi }}
		{\prfassumption{ \Delta, x^{\sigma} ; \Gamma, \xi: \Phi \vdash N: \Psi}}
		{\Delta ; \Gamma \vdash \letc \ [x, \xi ]:= M \ \inc \ N: \Psi} \\[0.2cm]
    
	\prftree
	    {\prfassumption{(\Delta ; \Gamma) \ \wfp}}
		{\prfassumption{ \Delta \vdash_{\T} t: \Nat}}
		{\Delta ; \Gamma \vdash \refle \, t: t = t } 
					&   
	\prftree
		{\prfassumption{ \Delta ; \Gamma \vdash M: t = u}}
		{\prfassumption{ \Delta ; \Gamma \vdash N: \Phi[x^{\Nat}: = t]}}
		{\Delta ; \Gamma \vdash \peel^{t,u}(M,\hat{x}.\Phi, N): \Phi[x^{\Nat}:= u] } \\[0.2cm]
	
	\multicolumn{2}{c}{
		\prftree
			{\prfassumption {\Delta ; \Gamma \vdash M: \Phi[x^{\Nat}:= 0]}}
			{\prfassumption {\Delta ; \Gamma \vdash N: \forall x^{\Nat}. (\Phi \Rightarrow \Phi[x^{\Nat}:= \succe \, x^{\Nat}])}}
			{\prfassumption {\Delta \vdash_{\T} t: \Nat}}
			{\Delta ; \Gamma \vdash \Ind (\hat{x}. \Phi, M, N, t): \Phi[x:=t]}
	}
	\end{array}
	\end{displaymath}
	\caption{\label{fig2} Proof derivations in $\LHAw$}
	\end{figure}
This
syntax contains three different $\lambda$-abstractions: 
two $\lambda$-abstractions at the level of proof terms ($\lambda \xi.M$ and
$\lambda x^{\sigma}.M$)
and the $\lambda$-abstraction of System~\T{} at the level
of formulas ($\lambda x^{\sigma}.t$). The sort annotation on a variable may be
omitted in the sequel if it can be inferred.
In the proof terms
$\peel^{t,u}(M,\hat{x}.\Phi, N)$ and $\Ind (\hat{x}. \Phi, M, N, t)$,
the variable $x$ is bound in $\Phi$: the binder $\hat{x}$ is used to specify
which variable will be substituted. The connectives 
$\top$ and $\vee$ are not included in \LHAw{} but can be
defined as $\top \equiv \bot \Rightarrow \bot$ 
and~$\Phi \vee \Psi \equiv \exists x^{\Nat} \,
(x = 0 \Rightarrow \Phi \wedge x \neq 0 \Rightarrow \Psi)$
where the relation $x \neq y$ denotes~$x = y \Rightarrow \bot$.

We consider sequents of the form
		$\Delta ; \Gamma \vdash M:\Phi$
	where 
	\begin{enumerate}
		\item $\Delta$ is a signature of System~\T{};
		\item $\Gamma$ is a context of \LHAw{}.
	\end{enumerate}

The typing rules of \LHAw{} are presented in Figure~\ref{fig2} at
page~\pageref{fig2}.
Note that equality is only defined on the sort $\Nat$.
A pair of a signature and a context $(\Delta;\Gamma)$
is well formed when the free first-order variables of 
$\Gamma$~are contained in $\Delta$, i.e. 
$$\begin{array}{rcrcl}
     (\Delta;\Gamma) \ \wfp & \equiv &
     \FV(\Gamma) & \subseteq & \Delta
    \end{array}$$
This system is not equipped with
reduction rules for proof terms: they are used
as annotations for the derivation and they serve as a tool
to formulate our work as a fully specified translation.
The congruence relation
    $ \Phi \simeq \Psi$
between formulas used in \LHAw{} is generated from the reduction
rules of System~\T{} and two extra rules:
$$\begin{array}{rcl}
         \nullt(0) & \succ & \top  \\
         \nullt(\succe \, x) & \succ & \bot
    \end{array}$$
Because of the conversion rule, terms of System~\T{} are treated up to the equivalence $\cong$. For instance, one
can prove $(\lambda x^{\Nat}.x) 0 = 0$ in \LHAw{}.
Moreover, all the axioms of \HAw{}~\cite{tro73}
are derivable. In particular,
the predicate $\nullt (t)$ is used to prove 
    $\forall x^{\Nat} \ \succe \, x \neq 0 $.

\begin{fact}
\LHAw{} captures \HAw{} in the following manner:
a closed formula $\Phi$ is derivable in \LHAw{}
if and only if the formula obtained from $\Phi$
by replacing all occurrences of  subformulas
$\nullt(t)$ by $t=0$ is a logical consequence
of \HAw{}.
\end{fact}

The notion of substitutions extends from System~\T{} to
\LHAw{}. Concretely, a (first-order)
substitution~$\theta$ is a finite function from first-order
variables ($x,y$...) to terms of System~\T{}, while the
operation of substitution on proof terms $M$ and
formulas $\Phi$ is defined as before. The notation
$\Gamma[\theta]$ represents the application of the
substitution $\theta$ to all terms and formulas in
the context $\Gamma$.
The system $\LHAw$ satisfies the following properties:

\begin{fact} 
\label{fact3}
If~$\Delta ; \Gamma  \vdash  M: \Phi$~then~ $FV(\Phi)  \subseteq  \Delta$.
\end{fact}

\begin{fact}
A generalized version of the weakening rule is admissible for
this system: 
    \begin{center}
        \begin{tabular}{crcl}
             if & $\Delta \subseteq \Delta'$,  
             $\Gamma \subseteq \Gamma'$
             and 
             $\Delta ; \Gamma \vdash M: \Phi$ & then &
             $\Delta' ; \Gamma' \vdash M: \Phi$
        \end{tabular}
    \end{center}
    where the set-theoretic inclusion is used to compare
    signatures and contexts.
\end{fact}

\begin{fact}
    Let $\theta$ be a substitution of
    first-order variables, $\Delta$ a signature
    included in its domain and $\Delta'$ a signature
    containing all free variables of its image.
    Then
        \begin{center}
        \begin{tabular}{rcl}
             $\Delta ; \Gamma \vdash M: \Phi$ & implies &
             $\Delta' ; \Gamma[\theta] \vdash 
             M[\theta]: \Phi[\theta]$
        \end{tabular}
    \end{center}
\end{fact}

\begin{fact} If $\Delta ; \Gamma , \xi:\Psi \vdash M: \Phi $ 
             and 
             $\Delta ; \Gamma \vdash N: \Psi$  then 
             $\Delta ; \Gamma \vdash M[\xi:= N]: \Phi$.
\end{fact}

\section{A preliminary study of possible extensions of equality}
\label{sec:PEE}

\subsection{Two examples of Partial Equivalence Relation}
\label{subsec:ExPER}

Let $\sigma$ be a sort of System~\T{}. 
A symbol of binary relation $\mathcal{R}$ on $\sigma$ (added to the syntax of $\LHAw$) is a partial equivalence relation
when it is symmetric and transitive. It is the case if the formulas
$$
\begin{array}{rcl}
     \Sym{}_{\mathcal{R}} & \equiv & 
     \forall x^{\sigma} \forall y^{\sigma} \ 
        x\mathcal{R} y \Rightarrow y \mathcal{R} x \\
   \Trans{}_{\mathcal{R}} & \equiv & 
     \forall x^{\sigma} \forall y^{\sigma} \forall z ^{\sigma} \ 
        x\mathcal{R} y \Rightarrow y \mathcal{R} z
          \Rightarrow x \mathcal{R} z
\end{array}
$$
are provable in $\LHAw$.

A partial equivalence relation is an equivalence relation on its
domain
$$
\begin{array}{rcl}
     x \in \Dom{}_{\mathcal{R}} & \equiv & x \mathcal{R} x
\end{array}
$$
Moreover, using symmetry and transitivity, one can show that 
$$x \mathcal{R} y   \Rightarrow 
     x \in \Dom{}_{\mathcal{R}} \wedge y \in \Dom{}_{\mathcal{R}}$$
Therefore, a partial equivalence relation on $\sigma$ 
is exactly an equivalence relation on a collection 
of individuals of sort $\sigma$
(i.e. a formula with one free variable of sort $\sigma$).

Let $\{\Ext_{\sigma}\}_{\sigma}$ and $\{\Eqpm_{\sigma}\}_{\sigma}$
be two families of binary relations indexed by the sorts of
System~\T{} and defined as follows:
$$
    \begin{array}{rcl@{\hspace{1.5cm}}rcl}
         x^{\Nat} \Ext_{\Nat} y^{\Nat} & \equiv & x = y &
         x^{\Nat} \Eqpm_{\Nat} y^{\Nat} & \equiv & x = y \\
         f^{\sigma \rightarrow \tau} \Ext_{\sigma \rightarrow \tau} 
         g^{\sigma \rightarrow \tau} & \equiv & 
         \forall x \ f \, x \Ext_{\tau} g \, x &
         f^{\sigma \rightarrow \tau} 
         \Eqpm_{\sigma \rightarrow \tau} 
         g^{\sigma \rightarrow \tau} & \equiv & 
         \forall x \forall y \ 
         x \Eqpm_{\sigma} y \Rightarrow
         f \, x \Eqpm_{\tau} g \, y
    \end{array}$$
Note that
\begin{enumerate}
    \item The relation \Ext{} is obtained from equality by 
    extending it to higher sorts in an extensional fashion 
    (two functions are in \Ext{} if they are extensionally equal).
    \item The relation \Eqpm{} is obtained from equality by
    extending it to higher sorts in a parametric fashion 
    (two functions are in \Eqpm{} if they send related
    entries to related outputs).
\end{enumerate}

With an (external) induction on the sorts of System~\T{}, it can
be shown that for all sorts $\sigma$:
\begin{enumerate}
    \item $\Ext_{\sigma}$ is an equivalence relation;
    \item $\Eqpm_{\sigma}$ is a partial equivalence relation.
\end{enumerate}
We exhibit proof terms
that are used on forthcoming translations:
$$\begin{array}{rccl}
        \vdash & \sympm_{\sigma}&:&   \Sym_{\Eqpm_{\sigma}} \\
        \vdash & \transpm_{\sigma} &: & \Trans_{\Eqpm_{\sigma}} \\
        \vdash & \reflpm_{\sigma} &: &
        \forall x^{\sigma} \forall y^{\sigma} \ x \Eqpm_{\sigma} y
    \Rightarrow (x \Eqpm_{\sigma} x \wedge y \Eqpm_{\sigma} y)
\end{array}$$
They are defined by induction on the sort of System~\T{}:
    $$\begin{array}{rcl}
         \sympm_{\Nat}& \equiv & \lambda x \lambda y. \lambda \xi.
         \peel(\xi, \hat{z} . (z = x), \refle \, x) 
         \\
         \sympm_{\sigma \rightarrow \tau }& \equiv & 
         \lambda f \lambda g. \lambda \xi. \lambda x \lambda y. 
         \lambda \eta.
         \sympm_{\tau} (f \,y) (g\, x) (\xi \, y \, x \, 
          (\sympm_{\sigma} \, x \, y \, \eta))
         \\
         \transpm_{\Nat} & \equiv &
         \lambda x \lambda y \lambda z . \lambda \xi \lambda  \eta. 
         \peel( \eta , \hat{w}. x = w , \xi) 
         \\ 
         \transpm_{\sigma \rightarrow \tau} & \equiv &
         \lambda f \lambda g \lambda h . \lambda \xi \lambda \eta. 
         \lambda x \lambda y. \lambda \chi.  
         \trans_{\tau} (f \, x) (g \, y) (h \, y) 
         (\xi \, x \, y \, \chi) (\eta \, y \, y 
         (\trans_{\sigma} \, y \, x \, y 
         (\sympm_{\sigma} x \, y \, \chi) \chi)) \\
         \reflpm_{\sigma} & \equiv &
         \lambda x^{\sigma} \lambda  y^{\sigma} \lambda \xi.
         (\transpm_{\sigma} \, x \, y \, x \, \xi 
         (\sympm \, x \, y \, \xi),
         \transpm_{\sigma} \, y \, x \, y \, 
         (\sympm \, x \, y \, \xi)
         \xi)
    \end{array}$$
One cannot prove inside \LHAw{} that $\Eqpm_{\sigma}$
is reflexive for all sorts $\sigma$,
as it can be seen by working in~$\HRO$. Let 
$\Hquote \in \HRO_{(\Nat \rightarrow \Nat) \rightarrow \Nat}$
be an index for the identity function\footnote{$\Hquote$ is a 
functional that takes a function as argument and returns its source code}
and 
$p,q \in \HRO_{\Nat \rightarrow \Nat}$ 
two distinct indexes for the same total unary function.
Note that
$$\begin{array}{rcl}
         \HRO & \vDash & p \Eqpm_{\Nat \rightarrow \Nat} q  \\
         \HRO & \vDash & 
         \{\Hquote\}(p) \neq_{\Nat}^{\mathrm{pm}} \{\Hquote\}(q)
    \end{array}$$
    Consequently
$$\begin{array}{rcl} 
         \HRO & \vDash & 
 \Hquote \neq^{\mathrm{pm}}_{(\Nat \rightarrow \Nat) \rightarrow \Nat} 
 \Hquote
    \end{array}$$
and $\Eqpm_{(\Nat \rightarrow \Nat) \rightarrow \Nat}$ is not
reflexive in~\HRO{}.

Because $\Ext_{(\Nat \rightarrow \Nat) \rightarrow \Nat}$ is 
reflexive, the previous result shows that in $\HRO{}$, 
$\Ext_{(\Nat \rightarrow \Nat) \rightarrow \Nat}$ 
is not included in 
$\Eqpm_{(\Nat \rightarrow \Nat) \rightarrow \Nat}$. Therefore, one cannot 
prove in~\LHAw{}
    $$\forall x \forall y \ 
    x \Ext_{(\Nat \rightarrow \Nat) \rightarrow \Nat} y
    \Rightarrow
    x \Eqpm_{(\Nat \rightarrow \Nat) \rightarrow \Nat} y$$
Going one step higher in the hierarchy of sorts, one can show that
$$\forall x \forall y \     x 
\Eqpm_{((\Nat \rightarrow \Nat) \rightarrow \Nat) \rightarrow \Nat } 
    y 
    \Rightarrow
    x 
\Ext_{((\Nat \rightarrow \Nat) \rightarrow \Nat) \rightarrow \Nat} 
    y$$
is not provable in \LHAw{}. Indeed, consider a variant \HROo{} of 
\HRO{} where natural numbers denote recursive functions 
that can access an oracle deciding if its entry is an index
of the identity function (i.e. for instance it 
returns $1$ if it is the case and $0$ otherwise). 
Let~$n \in 
    \HROo_{((\Nat \rightarrow \Nat) \rightarrow \Nat) \rightarrow
    \Nat}$
an index for this oracle and
    $m \in 
    \HROo_{((\Nat \rightarrow \Nat) \rightarrow \Nat) \rightarrow
    \Nat}$
an index of the constant function $x \mapsto 0$. It turns out
that
$$\begin{array}{rcl}
         \HROo & \vDash & n 
\Eqpm_{(\Nat \rightarrow \Nat) \rightarrow \Nat) \rightarrow \Nat} 
         m  \\
         \HROo & \vDash & 
n 
\neq^{\mathrm{ext}}_{(\Nat \rightarrow \Nat) \rightarrow \Nat) \rightarrow \Nat} 
         m
    \end{array}$$
because the indexes of the identity function are not in the domain of
$\Eqpm_{(\Nat \rightarrow \Nat) \rightarrow \Nat}$.

Finally, in the set-theoretic model \M{} of \HAw{} (and in fact
in any extensional model), one can
show
$$\begin{array}{rcl}
         \M & \vDash & \forall x \forall y \ x \Ext_{\sigma} y 
         \Leftrightarrow 
         x \Eqpm_{\sigma} y
    \end{array}$$
for all $\sigma$  (by induction on the sorts of System~\T{}). 

Therefore, in \LHAw{}, one cannot prove that the relations
$\Eqpm$ and $\Ext$ are different.
We wrap up all these results in the following
theorem.
\begin{thm}
In \LHAw{} one cannot prove that
\begin{enumerate}
    \item $\Eqpm_{(\Nat \rightarrow \Nat) \rightarrow \Nat}$ is reflexive;
    \item $\Ext_{(\Nat \rightarrow \Nat) \rightarrow \Nat} 
    \subseteq
    \Eqpm_{(\Nat \rightarrow \Nat) \rightarrow \Nat}$;    
    \item $
\Eqpm_{((\Nat \rightarrow \Nat) \rightarrow \Nat) \rightarrow \Nat} 
    \subseteq
\Ext_{((\Nat \rightarrow \Nat) \rightarrow \Nat) \rightarrow \Nat}$;
    \item $\Eqpm_{\sigma} \subsetneq \Ext_{\sigma}$
    and $\Ext_{\sigma} \subsetneq \Eqpm_{\sigma}$ (for any sort $\sigma$);
\end{enumerate}
where the symbols $\subseteq, \subsetneq$ and the property
of being reflexive are defined inside \LHAw{} in the usual way.
\end{thm}

\subsection{A first translation: from System~\T{} 
into \LHAw{}}
\label{subsec:FT}

Although one cannot prove inside \LHAw{} that \Eqpm{} is reflexive,
it can be shown that all closed terms of System~\T{} are in its
domain. With this goal in mind, we design a translation from 
System~\T{} into \LHAw{}:
$$\begin{array}{rcl}
             (\Delta \vdash_{\T} t: \sigma)^{\Pm} & \leadsto &  
    \Delta^1 , \Delta^2 ; 
    \Delta^{\Pm} \vdash t^{\Pm}: t^1 \Eqpm_{\sigma} t^2
        \end{array}$$
\begin{enumerate}
    \item Declarations of variables in signatures are duplicated.
    Fixing $i=1,2$:
$$\begin{array}{rcl}
             \emptyset^{i}& \equiv & \emptyset  \\
             (\Delta, x^{\sigma})^{i} & \equiv & 
             \Delta^{i}, (x^i)^{\sigma}
        \end{array}$$
where $x^i$ are fresh distinct variables.
    \item Terms of System~\T{} are duplicated. Fixing $i= 1,2$:
    $$\begin{array}{rcl}
             t^{i}& \equiv & t [\theta_t^i]
        \end{array}$$
    will denote the term obtained by substituting all free variables
    $x$ of $t$ by $x^i$ (i.e. $\theta_t^i$ is the 
    substitution defined on the free variables of
    $t$ that associates to a variable $x$
    the variable $x^i$).
    \item Signatures of System~\T{} are translated into contexts
    of \LHAw{}:
    $$\begin{array}{rcl}
             \emptyset^{\Pm}& \equiv & \emptyset  \\
             (\Delta, x^{\sigma})^{\Pm} & \equiv & 
             \Delta^{\Pm}, x^{\Pm}: x^1 \Eqpm_{\sigma} x^2
        \end{array}$$
    \item Terms of System~\T{} are translated into proof terms
    of \LHAw{}:
$$\begin{array}{rcl}
             (x)^{\Pm}& \equiv & x^{\Pm}  \\
              (\lambda x.t)^{\Pm} & \equiv &
              \lambda x^1 \lambda x^2. \lambda x^{\Pm}. t^{\Pm} \\
              (t \, u)^{\Pm} & \equiv & 
              t^{\Pm} \, u^1 \, u^2 \, u^{\Pm} \\
              0^{\Pm} & \equiv & \refle \, 0 \\
              (\succe \, t)^{\Pm} & \equiv & 
         \peel \bigl(t^{\Pm}, \hat{x}.(\succe \, t^1 = \succe \, x),
              \refle \, (\succe \, t^1)\bigr) \\
              (\Rec \, t \, u \, v)^{\Pm} & \equiv & 
              \Ind \bigl( 
              \hat{x}.( \forall y^{\Nat} \ x = y \Rightarrow
              \Rec^{\sigma} \, t^1 \, u^1 \, x
                \Eqpm_{\sigma}
                        \Rec^{\sigma} \, t^2 \, u^2 \, y), \\
                & & \hspace{1cm}
                \lambda y. \lambda \xi. \peel ( \xi, 
                \hat{z}.( t^1 
                \Eqpm_{\sigma} 
                \Rec^{\sigma} \, t^2 \, u^2 \, z), t^{\Pm}), \\
                & & \hspace{1cm}
                \lambda x . \lambda \eta.  \lambda y.
                \lambda \xi .
                \peel ( \xi , 
                \hat{z}. (u^1 (\Rec \, t^1 \, u^1 \, x) x 
                \Eqpm_{\sigma}
                (\Rec \, t^2 \, u^2 \, z)), \\
              & & \hspace{3.8cm}
                 u^{\Pm} (\Rec \, t^1 \, u^1 \, x)
                (\Rec \, t^2 \, u^2 \, x) 
                (\eta \, x \, (\refle \, x)) \, x \, x \, 
                (\refle \, x)),
                \\
                & & \hspace{1cm}
                v^1) v^2 \, v^{\Pm}
        \end{array}$$
\end{enumerate}
The translation works as follows:
\begin{enumerate}
    \item An abstraction in System~\T{} is interpreted as
    3 abstractions in \LHAw{}: 2 abstractions of first-order variables
    and one of proof variable. Indeed, 
    $(\lambda x^{\sigma}.t)^{\Pm}$ is of type
    ${\forall x^1 \forall x^2 \, x^1 \Eqpm_{\sigma} x^2 \Rightarrow
    t^1 \Eqpm t^2}$ (assuming $\lambda x^{\sigma}.t: \sigma \rightarrow 
    \tau)$.
    \item Symmetrically, an application in System~\T{} is interpreted
    as 3 applications.
    \item Because the relation $\Eqpm_{\Nat}$ is merely the equality,
    $0$ and $\succe \, t$ are interpreted as equality proofs.
    \item Finally, the recursor is interpreted using an induction.
    During the induction, the synchronization between the two
    copies of the term $v$ (of sort $\Nat$) is lost. Therefore,
    we need an extra generalization in the hypothesis to retrieve
    that these terms are equal.
\end{enumerate}

\begin{thm} If
        $\Delta   \vdash_{\T}  t: \sigma$
then
        $\Delta^1,\Delta^2 ; \Delta^{\Pm}   \vdash  t^{\Pm}: 
        t^1 \Eqpm_{\sigma} t^2$.
In particular, $\vdash t^{\Pm}: 
        t \Eqpm_{\sigma} $t
for all closed terms of sort $\sigma$.
\end{thm}

\begin{proof}
    By induction on the derivations of System~\T{}.
\end{proof}

We deduce from the previous theorem that each
closed term of System~\T{} is in the domain of
\Eqpm{}. 

The following terms are used later:
$$\begin{array}{rcl}
        \Delta^1,\Delta^2 ; \Delta^{\Pm}
        & \vdash & 
        \Elim^i_{\hat{z}.t}: 
        \forall z^1 \forall z^2 \ z^1 \Eqpm_{\sigma} z^2
        \Rightarrow 
        t^i[z^i = z^1] \Eqpm_{\tau} 
        t^i[z^i = z^2]
    \end{array}$$
for $i=1,2$. Note that these proof terms are indexed by
a term $t$ and a variable $z$.
These terms are constructed using the previous translation,
as follows:
$$\begin{array}{rcl}
         \Elim^1_{\hat{z}.t} & \equiv &
         \lambda z^1 \lambda z^2 . \lambda z^{\Pm}.
        \transpm \, t^1 \, t^2 t^1[z^1:= z^2] t^{\Pm} \\
        & & \hspace*{2cm}
        (\sympm t^1[z^1:= z^2] t^2
        t^{\Pm}[z^1:= z^2][z^{\Pm}:= (\reflpm z^1 z^2 z^{\Pm}).2]
        )
        \\ 
         \Elim^2_{\hat{z}.t} & \equiv &
         \lambda z^1 \lambda z^2 . \lambda z^{\Pm}.
        \transpm \, t^2[z^2:= z^1] \, t^1 t^2 \\
        & & \hspace*{2cm}
        (\sympm t^1 t^2[z^2:= z^1] 
        t^{\Pm}[z^2:= z^1] 
        [z^{\Pm}:= (\reflpm \, z^1 z^2 z^{\Pm}).1] ) t^{\Pm})
        \\
    \end{array}$$
They are well typed as soon as $\FV(t) \subseteq \Delta,z$.

\section{Interpreting extensional equality: from \LEHAw{} into \LHAw{}}
\label{sec:ST}

\subsection{A preliminary step: a translation from \LHAw{}
into \LHAw{}}
\label{subsec:Intparam}
Although all closed terms of System~\T{} are in the domain of \Eqpm{},
one cannot prove
$\forall x^{\sigma} \ x \Eqpm_{\sigma} x$
in $\LHAw$ (for $\sigma = (\Nat \rightarrow \Nat) \rightarrow \Nat$
for instance). Nevertheless, building on the intuition of
restricting quantifications and
on the work of the previous section, we can
design a translation 
$$\begin{array}{rcl}
             (\Delta ; \Gamma \vdash M: \Phi)^{\Pm} & \leadsto &  
    \Delta^1 , \Delta^2 ; 
    \Delta^{\Pm} , \Gamma^{\Pm} \vdash M^{\Pm}: \Phi^{\Pm}
        \end{array}$$
from $\LHAw$ into itself that will serve
as a basis to interpret extensional equality.
\begin{enumerate}
    \item This translation extends the one defined formerly.
In particular 
$\Delta^i, \Delta^{\Pm}, t^i \mbox{ and } t^{\Pm}$
are already defined.
    \item Formulas of
    $\LHAw$ are translated into formulas of $\LHAw$:
 $$\begin{array}{rcl}
             (t = u)^{\Pm} & \equiv & t^1 \Eqpm_{\Nat} u^2 \\  
             \bot^{\Pm} & \equiv & \bot \\
             (\Phi \Rightarrow \Psi)^{\Pm} & \equiv &
             \Phi^{\Pm} \Rightarrow \Psi^{\Pm} \\
             (\Phi \wedge \Psi)^{\Pm} & \equiv &
             \Phi^{\Pm} \wedge \Psi^{\Pm} \\
             (\forall x^{\sigma} \Phi)^{\Pm} & \equiv &
             \forall x^1 \forall x^2 
             \ x^1 \Eqpm_{\sigma} x^2 \Rightarrow \Phi^{\Pm} \\
             (\exists x^{\sigma} \Phi)^{\Pm} & \equiv &
             \exists x^1 \exists x^2 
             \ x^1 \Eqpm_{\sigma} x^2 \wedge \Phi^{\Pm}
        \end{array}$$
    \item Contexts of $\LHAw$ are translated into contexts of 
    $\LHAw$:
    $$\begin{array}{rcl}
             \emptyset^{\Pm}& \equiv & \emptyset  \\
             (\Gamma, \xi: \Phi)^{\Pm} & \equiv & 
             \Gamma^{\Pm}, \xi: \Phi^{\Pm}
        \end{array}$$
    \item Proof terms of $\LHAw$ are translated into proof terms of 
    $\LHAw$:
    $$\begin{array}{rcl}
             (\xi)^{\Pm}& \equiv & \xi  \\
             (\lambda \xi. M)^{\Pm} & \equiv & \lambda \xi. M^{\Pm} 
             \\
             (M \,N)^{\Pm} & \equiv & M^{\Pm} N^{\Pm} \\
             (M,N)^{\Pm} & \equiv & (M^{\Pm},N^{\Pm}) \\
             (M.i)^{\Pm} & \equiv & M^{\Pm}.i \\
             (\lambda x. M)^{\Pm} & \equiv & 
             \lambda x^1 \lambda x^2 \lambda x^{\Pm}. M^{\Pm} 
             \\
             (M \, t)^{\Pm} & \equiv & 
             M^{\Pm} t^1 t^2 t^{\Pm} \\
             ([t,M])^{\Pm} & \equiv & 
             [t^1,[t^2,(t^{\Pm},M^{\Pm})]] 
             \\
             (\letc \, [x,\xi]:= M \, \inc \, N)^{\Pm} & \equiv & 
             \letc \, [x^1,\eta]:= M^{\Pm} \, \inc \,
             \letc \, [x^2, \chi]:= \eta \, \inc \, 
             N^{\Pm}[x^{\Pm}=\chi.1][\xi^{\Pm}:= \chi.2]
             \\
             (\efq(M,\Phi))^{\Pm} & \equiv & 
             \efq(M^{\Pm},\Phi^{\Pm}) \\
             (\refle \, t)^{\Pm} & \equiv & t^{\Pm} \\
             \bigl(\peel^{t,u}(M,\hat{x}.\Phi,N)\bigr)^{\Pm} 
             & \equiv &
             \peel\bigl
             (M^{\Pm}_2 , \hat{x^2}. \Phi^{\Pm}[x^1:= u^1],
             \peel(M^{\Pm}_1 , \hat{x^1}.\Phi^{\Pm}[x^2:= t^2],
             N^{\Pm})\bigr) \\
             \bigl(\Ind(\hat{x}.\Phi,M,N,t)\bigr)^{\Pm} & \equiv &
             \Ind(\hat{x}. \forall y \ x = y \Rightarrow
             \Phi^{\Pm}[x^1:= x][x^2:=y], \\
             & & \hspace*{0.5cm}
             \lambda y \lambda \xi. 
             \peel (\xi, \hat{z}.\Phi^{\Pm}[x^1:=0][x2:=z] ,
             M^{\Pm}),\\
             & & \hspace*{0.5cm}
             \lambda x \lambda \eta \lambda y \xi.
             \peel 
             (\xi, \hat{z}.\Phi^{\Pm}[x^1:=\succe \, x][x^2:=z], 
             N^{\Pm} x \, x \,(\refle \, x) 
             (\eta \, x \, (\refle \, x)) , \\
             & & \hspace{0.5cm}
             t^1) 
             t^2 t^{\Pm}
        \end{array}$$
where in the translation of $\peel(M,\hat{x}.\Phi,N)$, 
    $M^{\Pm}_i$ denotes a proof of $t^i = u^i$: 
    $$\begin{array}{rclcl}
             M^{\Pm}_1 & \equiv & 
             \transpm_{\Nat} t^1 \, u^2 \, u^1 \, M^{\Pm} 
             (\sympm_{\Nat} u^1 \, u^2 u^{\Pm}) &: & t^1 = u^1
             \\ 
             M^{\Pm}_2 & \equiv & 
             \transpm_{\Nat} t^2 \, t^1 \, u^2  
             (\sympm_{\Nat} t^1 \, t^2 t^{\Pm}) 
             M^{\Pm} &: & t^2 = u^2
        \end{array}$$
\end{enumerate}
Here, the translation of $\peel$ is ad hoc: it is merely done
by using $\peel$ on two distinct equalities. However,
it will not be the case in the last translation, where we 
will need an 
external recursion on the formulas of
the source system to interpret it.

The translation of induction follows the same principle as the
translation of the recursor done in 
Section~\ref{subsec:FT}: because the
synchronization between the copies of~$t$~is lost, we need to 
generalize the inductive hypothesis.

\begin{thm}
If
$\Delta ; \Gamma  \vdash  
M: \Phi$ ~then~$\Delta^1 , \Delta^2 ; 
         \Delta^{\Pm}, \Gamma^{\Pm}  \vdash 
         M^{\Pm}: \Phi^{\Pm}$. 
\end{thm}

The proof of this theorem is by induction on the derivation of
$\LHAw$ and it uses three lemmas:

\begin{lemma} 
If~$t  \cong  u $ then $t ^i \cong  u^i$~for $i=1,2$ 
and $t,u$ terms of System~\T{}. 
\end{lemma}

\begin{lemma}
If $t(x^{\sigma})$ and $u$ are terms 
with $u$ of sort $\sigma$, then~$\bigl(t[x:=u]\bigr)^i 
\equiv t^i[x^i:= u^i]$
for $i=1,2$.
\end{lemma}

\begin{lemma}
If $\Phi(x^{\sigma})$ is a formula and $t$ is a 
term of sort $\sigma$, then
         $\bigl(\Phi[x:=t]\bigr)^{\Pm}  \equiv 
         \Phi^{\Pm}[x^1:= t^1][x^2:= t^2 ]$.
\end{lemma}

\subsection{Extending equality through parametricity: 
a translation from
\LEHAw{} into \LHAw{}}
\label{subsec:ExtEq}

    \begin{figure}[t]
        \centering
    \begin{displaymath}
	\begin{array}{cc}
	\prftree
	    {\prfassumption{(\Delta ; \Gamma) \ \wfp}}
		{\prfassumption{ \Delta \vdash_{\T} t: \sigma}}
		{\Delta ; \Gamma \vdash_{\extp} 
		\refle_{\sigma} \, t: t =_{\sigma} t } 
					&   
	\prftree
		{\prfassumption{ \Delta ; \Gamma \vdash_{\extp} 
		M: t =_{\sigma} u}}
		{\prfassumption{ \Delta ; \Gamma \vdash_{\extp} 
		N: \Phi[x^{\sigma}: = t]}}
		{\Delta ; \Gamma \vdash_{\extp}
		\peel^{t,u}_{\sigma}(M,\hat{x}.\Phi, N): 
		\Phi[x^{\sigma}:= u] } \\[0.2cm]
		
	\prftree
	    {\prfassumption{ \Delta ; \Gamma \vdash_{\extp} M:
	    \forall x^{\sigma} f \, x =_{\tau} g \, x  }}
	    {\Delta ; \Gamma \vdash_{\extp} 
	    \extpm_{\sigma , \tau}(M): 
	    f =_{\sigma \rightarrow \tau} g } & 
	    
	\prftree
	    {\prfassumption{ \Delta ; \Gamma \vdash_{\extp} 
	    M: f =_{\sigma \rightarrow \tau} g}}
	    {\prfassumption{ \Delta ; \Gamma \vdash_{\extp} 
	    N: t =_{\sigma} u }}
	    {\Delta ; \Gamma \vdash_{\extp} 
	    \apppm_{\sigma , \tau}(M,t,u,N): 
	    f \, t=_{\tau} g \, u }
	\end{array}
	\end{displaymath}
\caption{\label{fig3} Additional typing rules for extensional
equality}
	\end{figure}

Our next goal is to extend the previous translation to give
an interpretation of extensional equality inside~\LHAw{}.

Consider an extension \LEHAw{} of \LHAw{} obtained by extending
equality in an extensional way to all higher sorts,
i.e. by adding
\begin{enumerate}
    \item atomic formulas $t =_{\sigma} u$ for all sorts $\sigma$;
    \item proof terms $(\refle_{\sigma} \, t)$,
    $\peel^{t,u}_{\sigma}(M,\hat{x}.\Phi,N)$,
    $\extpm_{\sigma , \tau}(M)$ and
    $\apppm_{\sigma , \tau}(M,t,u,N)$
    for all sorts~$\sigma$~and~$\tau$;
    \item typing rules for the added proof terms presented in
    Figure~\ref{fig3} at page~\pageref{fig3}.
 \end{enumerate}

The symbol 
$\vdash_{\extp}$
will be used to denote sequents (and provability) in $\LEHAw$.
The translation $(\_)^{\Pm}$ can be extended to a translation
from $\LEHAw$ into $\LHAw$.
Indeed, one interprets
$$\begin{array}{rcl}
         (t =_{\sigma} u)^{\Pm} & \equiv & 
         t^1 \Eqpm_{\sigma} u^2
    \end{array}$$
It is then possible to extend the translation with
$$\begin{array}{rcl}
         (\refle_{\sigma} \,t)^{\Pm} & \equiv & t^{\Pm}
    \end{array}$$
and still preserving adequacy. The case of
    $\bigl(\peel^{t,u}_{\sigma}(M,\hat{x}.\Phi,N)\bigr)^{\Pm}$
is more involved and it is treated as follows.
We first construct
a family of terms
    $\Elim_{\hat{x}.\Phi}$
satisfying that if
$\FV(\Phi)  \subseteq  \Delta,x^{\Pm}$
then
$$\begin{array}{rcl}
         \Delta^1 , \Delta^2 ; \Delta^{\Pm} 
         & \vdash & \Elim_{\hat{x}.\Phi}: 
         \forall x^1 \forall x^2 \forall y^1 \forall y^2 \ 
          x^1 \Eqpm_{\sigma} y^1 \Rightarrow 
          x^2 \Eqpm_{\sigma} y^2 \Rightarrow 
          \Phi^{\Pm} \Rightarrow \Phi^{\Pm}[x^1:= y^1][x^2:=y^2]
    \end{array}$$
This is done by induction on the syntax of formulas:
$$\begin{array}{rclr}
         \Elim_{\hat{x}.t =_{\sigma} u} & \equiv &  
         \lambda x^1 \lambda x^2 \lambda y^1 \lambda 
         y^2 \lambda \xi^1 \lambda \xi^2 \lambda \xi. 
        \transpm \, t^1[x^1:= y^1] \, t^1 \, u^2[x^2:= y^2] \\
        && \hspace*{4.8cm}
        (\Elim^1_{\hat{x}.t} \ y^1 \, x^1 \, 
        (\sympm \, x^1 \, y^1 \, \xi^1)) \\
        && \hspace*{4.8cm}
        (\transpm \, t^1 \, u^2 \, u^2[x^2:= y^2] 
        \xi
        (\Elim^2_{\hat{x}.u} \ x^2 \, y^2 \, \xi^2)) \\
        \Elim_{\hat{x}.\bot} & \equiv &  
         \lambda x^1 \lambda x^2 \lambda y^1 \lambda 
         y^2 \lambda \xi^1 \lambda \xi^2 \lambda \xi.
        \xi \\
        \Elim_{\hat{x}.(\Phi \Rightarrow \Psi)} & \equiv &  
        \lambda x^1 \lambda x^2 \lambda y^1 \lambda 
         y^2 \lambda \xi^1 \lambda \xi^2 \lambda \xi. 
         \lambda \eta. \Elim_{\hat{x}.\Psi}^+ 
         (\xi \ (\Elim^- \eta )) \\
         \Elim_{\hat{x}.(\Phi \wedge \Psi)} & \equiv &  
        \lambda x^1 \lambda x^2 \lambda y^1 \lambda 
         y^2 \lambda \xi^1 \lambda \xi^2 \lambda \xi.
         (\Elim_{\hat{x}.\Phi}^+ \xi.1 , 
         \Elim_{\hat{x}.\Psi}^+ \xi.2) \\
         \Elim_{\hat{x}.(\forall z \Phi)} & \equiv &
         \lambda x^1 \lambda x^2 \lambda y^1 \lambda 
         y^2 \lambda \xi^1 \lambda \xi^2 \lambda \xi.
         \lambda z^1 \lambda z^2 
         \lambda z^{\Pm}. \Elim_{\hat{x}.\Phi} (\xi \, z^1 z^2 z^{\Pm})\\
         \Elim_{\hat{x}.(\exists z \Phi)} & \equiv &
         \lambda x^1 \lambda x^2 \lambda y^1 \lambda 
         y^2 \lambda \xi^1 \lambda \xi^2 \lambda \xi.
        \letc \, [z^1, \eta ]:= \xi \, \inc \, 
         \letc \, [z^2,\chi]:= \eta \, \inc \, 
         [z^1,[z^2, (\eta.1, \Elim_{\hat{x}.\Phi}^+ \eta.2) ]]
    \end{array}$$
where 
$$\begin{array}{rcl}
         \Elim_{\hat{x}.\Phi}^+ & \equiv & 
         \Elim_{\hat{x}.\Phi} \, x^1 \, x^2 \,
         y^1\,y^2\,\xi^1\,\xi^2\\
         \Elim^- & \equiv & 
         \Elim_{\hat{y}.\Phi[x:=y]} \, y^1 \, y^2 \,
         x^1\,x^2\, (\sympm \, x^1 \, y^1 \, \xi^1)
         (\sympm \, x^2 \, y^2 \, \xi^2)
    \end{array}$$
The proof that $\Elim_{\hat{x}.\Phi}$ satisfies 
the given property is by induction on the syntax, 
where the hypothesis is used to treat the case of equality.

We can now define
$$\begin{array}{rcl}
    \bigl(\peel^{t,u}_{\sigma}(M,\hat{x}.\Phi,N)\bigr)^{\Pm}
        & \equiv & 
        \Elim_{\hat{x}.\Phi} t^1 t^2 u^1 u^2
        (\transpm t^1 u^2 u^1 M^{\Pm}(\sympm u^1 u^2 u^{\Pm})) \\
        && \hspace*{2.5cm}
        (\transpm t^2 t^1 u^2 (\sympm t^1 t^2 t^{\Pm}) M^{\Pm}) 
        N^{\Pm}
    \end{array}$$
Finally, we set
$$\begin{array}{rcl}
    \bigl(\extpm_{\sigma \rightarrow \tau} (M)\bigr)^{\Pm}
        & \equiv & 
        \lambda x^1 , x^2 \lambda x^{\Pm}. 
        M^{\Pm} \, x^1 \, x^2 \, x^{\Pm} \\
    \bigl(\apppm_{\sigma \rightarrow \tau}(M,t,u,N)\bigr)^{\Pm}
    & \equiv &
        M^{\Pm} \, t^1 \, u^2 \, N^{\Pm}
    \end{array}$$
\begin{thm}
If
         $\Delta ; \Gamma  \vdash_{\extp}  M: \Phi 
$
then
         $\Delta^1 , \Delta^2 ; 
         \Delta^{\Pm}, \Gamma^{\Pm} \vdash  M^{\Pm}: 
         \Phi^{\Pm}$.
\end{thm}
\begin{proof}
The case of 
    $\bigl(\peel_{\sigma}(M,\hat{x}.\Phi,N)\bigr)^{\Pm}$
uses the property of
    $\Elim_{\hat{x}.\Phi}$
and a generalization of Fact \ref{fact3}, saying that if
$    \Delta ; \Gamma  \vdash_{\extp}  M: \Phi $
then
    $ \FV(\Phi)  \subseteq  \Delta.$
\end{proof}

\begin{cor}
If \LHAw{} is consistent, then so is \LEHAw{}.
\end{cor}

\begin{proof}
    If \LEHAw{} is inconsistent,
    there exists a proof term $M$ 
    such that $\vdash_e M:\bot$.
    By the previous translation, one gets
    a derivation of $\vdash M^{\Pm}: \bot$
    and concludes that \LHAw{} is inconsistent.
\end{proof}

\subsection{Characterizing the image of the translation}
\label{subsec:Char}
In the previous section, we showed that if a closed formula
$\Phi$ is provable in \LEHAw{} then $\Phi^{\Pm}$ is provable
in \LHAw{}. The goal of this section is to prove the converse:
if $\Phi^{\Pm}$ is provable in \LHAw{} then $\Phi$ is provable in
\LEHAw{}. These properties show
that the system \LEHAw{} fully characterizes the
image of the translation we designed. 

We first show that the relation $\Eqpm$ collapses to
the equality relation in \LEHAw{}.
For every sort $\sigma$, we construct a proof term 
$$\vdash_e \Collaps_{\sigma}: 
    \forall x^{\sigma} \forall y^{\sigma} \ 
    x =_{\sigma} y  \Leftrightarrow x \Eqpm_{\sigma} y$$
by external induction on the sorts of System~\T{}:
$$\begin{array}{rcl}
         \Collaps_{\Nat}& \equiv &
         \lambda x \lambda y \bigl(\lambda \xi.\xi, \lambda
         \xi.\xi)\\
         \Collaps_{\sigma \rightarrow \tau} & \equiv &
         \lambda f \lambda g
         (\lambda \xi \lambda x \lambda y \lambda \eta.
         \Collaps_{\tau}.1 \, (f \, x) \, (g \,y) \, 
         \apppm_{\sigma,\tau}
         ( \xi , x,y, \Collaps_{\sigma}.2 \, 
         x \, y \, \eta), \\ & & \hspace*{2.6cm} 
         \lambda \xi. \extpm_{\sigma,\tau} ( \lambda z. 
         \Collaps_{\tau}.2 \, (f \, z) \, (g \, z) 
         (\xi \, z \, z \,
         (\Collaps_{\sigma}.1 \, z \, z _, 
         (\refle \,z))))\bigr)
    \end{array}$$
\begin{prop}
\label{Prop:colaps}
    For every sort $\sigma$, the binary relations
    $\Eqpm_{\sigma}$ and $\Ext_{\sigma}$
    collapse to $=_{\sigma}$ in $\LEHAw{}$.
\end{prop}
\begin{proof}
    We proved above
    that $\Eqpm_{\sigma}$ collapses to $=_{\sigma}$
    and it can be proved in a similar fashion
    that $\Ext_{\sigma}$ also collapses
    to $=_{\sigma}$~in~$\LEHAw{}$.
\end{proof}

Using the Proposition~\ref{Prop:colaps}, 
we can now show that the image
of the last translation is fully characterized by the type
system \LEHAw{}.

We exhibit a family of proof terms $\Equiv^i_{\Phi}$
for $i = 1,2$ satisfying, 
for any formula $\Phi$ and any signatures~$\Delta$
containing the free variables of $\Phi$, that
$$\begin{array}{rcl}
         \Delta^1, \Delta^2 ; \Delta^{\Pm} & \vdash_e &  
         \Equiv^1_{\Phi}
        :
         \Phi^1 \Rightarrow \Phi^{\Pm} \\
        \Delta^1 , \Delta^2 ; \Delta^{\Pm} & \vdash_e &  
        \Equiv^2_{\Phi}
       :
         \Phi^{\Pm} \Rightarrow \Phi^1
    \end{array}$$
Such proof terms are defined as follows:
$$\begin{array}{rcl}        \Equiv^1_{t =_{\sigma} u} & \equiv &
        \lambda \xi . \trans_{\sigma} \, t^1 \, u^1 \, u^2
        (\Collaps_{\sigma}.1 \, t^1 \, u^1 \, \xi) 
        \, u^{\Pm} \\
        \Equiv^2_{t =_{\sigma} u} & \equiv &
        \lambda \xi. 
        \Collaps_{\sigma}.2 \, t^1 \, u^1 \, 
        (\trans_{\sigma} \, t^1 \, u^2 \, u^1 \,
        \xi \, (\sympm \, u^1 \, u^2 \, u^{\Pm}))
        \\[0.3cm]
        \Equiv^1_{\Phi \Rightarrow \Psi} & \equiv &
        \lambda \xi \lambda \eta. 
        \Equiv^1_{\Psi} (\xi \, (\Equiv^2_{\Phi} \, \eta)) \\
        \Equiv^2_{\Phi \Rightarrow \Psi} & \equiv &
        \lambda \xi \lambda \eta. 
        \Equiv^2_{\Psi} (\xi \, (\Equiv^1_{\Phi} \, \eta))
        \\[0.3cm]
        \Equiv^1_{\Phi \wedge \Psi} & \equiv &
        \lambda \xi. 
        (\Equiv^1_{\Phi} \xi.1 , \Equiv^1_{\Psi} \xi.2) \\
        \Equiv^2_{\Phi \wedge \Psi} & \equiv &
        \lambda \xi. 
        (\Equiv^2_{\Phi} \xi.1 , \Equiv^2_{\Psi} \xi.2 )
        \\[0.3cm]

        \Equiv^1_{\forall x^{\sigma} \, \Phi} & \equiv &
        \lambda \xi \lambda x^1 \lambda x^2. \lambda x^{\Pm}.
        \Equiv^1_{\Phi} \, (\xi \, x^1)
        \\
        \Equiv^2_{\forall x^{\sigma} \, \Phi} & \equiv &
        \lambda \xi \lambda x^1 . 
        \Equiv^2_{\Phi}[x_2:= x_1]
        [x^{\Pm}:= (\Collaps_{\sigma}.1 \, x^1 \, x^1 \,
        (\refle_{\sigma} \, x^1))] \\
        & & \hspace{2cm}
        \, (\xi \, x^1 \, x^1 \, 
        (\Collaps_{\sigma}.1 \, x^1 \, x^1 \,
        (\refle_{\sigma} \, x^1))
        )
        \\[0.3cm]
        \Equiv^1_{\exists x^{\sigma} \Phi} & \equiv &
        \lambda \xi. \letc \, [x,\eta]:= \xi \, \inc \, 
        [x,[x,(\Collaps_{\sigma}.1 \, x \, x \,
        (\refle_{\sigma} \, x), \Equiv^1_{\Phi} \xi)]] \\
        \Equiv^2_{\exists x^{\sigma} \Phi} & \equiv &
        \lambda \xi. 
        \letc \, [x^1,\eta]:= \xi \, \inc \,
        \letc \, [x^2,\chi]:= \eta \, \inc \,
        [x^1 , \Equiv^2_{\Phi} \chi.2 ] 
\end{array}$$
We can then conclude that a closed formula $\Phi$ is provable in
$\LEHAw{}$ if and only if its translation
is provable in $\LEHAw{}$.
\begin{thm}
For a closed formula $\Phi$ of $\LEHAw{}$
$$
         \vdash_e (\Equiv^1_{\Phi},\Equiv^2_{\Phi}): 
         \Phi \Leftrightarrow \Phi^{\Pm}
    $$
In particular, if $\Phi^{\Pm}$ is provable in \LHAw{}
then $\Phi$ is provable in \LEHAw{}.
\end{thm}

\subsection{Adding reduction rules: a conjecture}
The proof systems used here lack of computational rules,
such as
$$\begin{array}{rcl}
         (\lambda x.M) \, t & \succ_{\beta} & M[x:=t] \\
         (\lambda \xi.M) \, N & \succ_{\beta} & M[\xi:=N] \\
         \letc \, [x,\xi]:= [t,M] \, \inc \, N & \succ_{\beta} &
         N[x:=t][\xi:= M] \\
         (M_1,M_2).i & \succ_{\beta} & M.i \\
         \Ind(\hat{x}.\Phi,M,N,0) & \succ_{\iota} & M \\
         \Ind(\hat{x}.\Phi,M,N,S \, t) & \succ_{\iota} & 
         N \, t \, \Ind(\hat{x}.\Phi,M,N,t)\\
        \peel(\refle \, t, \hat{x}.\Phi,N) & \succ_{\iota} & N 
    \end{array}$$
and in $\LEHAw$
$$\begin{array}{rcl}
        \apppm(\extpm(M),t,t,\refle \, t) & \succ &  
        M \, t
    \end{array}$$
One could try to figure out if the translation $(\_)^{\Pm}$
respects reductions, i.e. a property of the
shape~$M \leadsto N$  implies 
         $M^{\Pm} \simeq_{\beta,\iota} N^{\Pm}$.
While this is true for $\beta$-reductions, 
it seems that it does not hold for the reduction
of $ \Ind(\hat{x}.\Phi,M,N, \succe \, t)$ if $t$ contains free first-order variables.
Indeed, the term $\succe \,t$ will be translated as a proof term 
asserting an equality and if it does not compute into
$\refle \, (\succe \,t)$, 
it will block the reduction of the subterm $\peel(...)$
inside $(\Ind(\hat{x}.\Phi,M,N,S \, t))^{\Pm}$. In the case of
a closed term $t$, we conjecture 
it will compute as desired.
Concretely, we think
that the proof will use the meta properties of System~\T{} described
at the end of Section~\ref{subsec:SystemT}: one will use that
every closed term of sort $\Nat$ is $\beta$-equivalent to a term
of the shape $\succe^n \, 0$ and that the translation of such a 
term computes into $\refle \, (\succe^n \, 0)$, assuming that
the translation from System~\T{} to \LHAw respects reductions.

\begin{conj} 
If $M$ is a proof term of
\LEHAw{} without free first-order variables
and if $M$ does not contain $\peel$ then
$M \leadsto N$  implies 
         $M^{\Pm} \simeq_{\beta,\iota} N^{\Pm}$.
\end{conj}

However, 
the case of the reduction rule of 
$\peel$ seems more difficult to treat, as 
the translation of $\peel$ relies on an
external induction over the syntax but also because it uses
proofs of symmetry and transitivity of \Eqpm{} that do not 
seem to compute as needed. 

\section{Related work}
\label{sec:RW}

The idea to build
a syntactic model satisfying extensionality axioms
is already present in Gandy's work~\cite{gan56}. 
Concretely, in higher-order
logic, Gandy defined a syntactic model by restricting
the elements of discourse to \emph{parametric}
ones and proved that, in this model, 
extensionally equal elements satisfy the
same properties. 
Here, we adapt this construction
to the theory \HAw{} and, using ideas of 
the Curry-Howard correspondence, we formulate it
as a translation of proof systems. 
Our translation
slightly differs from the one of Gandy because
we use techniques from parametricity.
This choice comes from the idea that parametric 
translation can serve to extend equality.
Nevertheless, because extensionality and parametricity
relations collapse to equality in an extensional model
(as shown in Proposition~\ref{Prop:colaps}), it 
is somehow a matter of design:
the translation from~$\LEHAw{}$ into~$\LHAw{}$
could be designed without the use of parametricity.

Zucker already proved a result of
relative consistency 
between \EHAw{} and \NHAw{}~\cite{zuc71,tro73}.
He did it in a semantical fashion by transforming
models of \NHAw{} into models of \EHAw{}. The
method he used is similar to the method of Gandy
but, in this context, it suffices to restrict
the domain to \emph{parametric} elements and to check that 
the relation of extensionality is an equivalence
relation that is congruent (thus suited to
interpret equality).
In our work, rather than from \NHAw{}, we start with
\HAw{}: we reconstruct the equality
from scratch and show that it respects
Leibniz principle.
Despite this slight difference, our
work can be seen as the syntactical counterpart
of the result of Zucker. 
One advantage is that a syntactical
translation comes with an explicit translation
of proofs, that we formulate here as a translation
of proof terms.

The ideas behind the proof system \LHAw{} are folklore.
For instance, representing the axiom scheme of 
induction as an inference rule can be seen in
many other proof systems, as for instance in Martin-L\"of Type Theory~\cite{MLTT84}. The idea 
to use the predicate $\nullt(t)$ to derive
Peano's fourth axiom
already appears in Miquel's work~\cite{Miq09}.
The terminology~"$\peel$" that we use to denote
the eliminator of equality is similar to the
one used in some presentations of type theory~\cite{NBP90};
however our own motivation to use it is
to emphasize
that Leibniz principle is recovered
by doing an external induction on the formulas or,
more graphically, by peeling out the syntax.

\section{Conclusion}
\label{sec:Conc}

We designed a translation from \LEHAw{} into \LHAw{} using
techniques reminiscent of parametricity and we proved a result of relative consistency: if $\LHAw{}$ is consistent, then so is \LEHAw{}. 
The following diagram shows an intuition of the translation:

    \begin{center}
        \begin{tabular}{c@{\hspace*{2cm}}c@{\hspace*{2cm}}c}
        \LEHAw{} &  &  \LHAw{} \\
        \hline
             $t$ & $\leadsto$  & 
    \begin{tikzcd}
t^1 
\ar[equal]{d}{t^{\mathrm{pm}}} \\
t^2
\end{tikzcd} \\
$M: t = u$  & $\leadsto$ & 
\begin{tikzcd}
t^1 \ar[equal,d, "t^{\mathrm{pm}}"'] \arrow[equal,rd, " M^{\mathrm{pm}}",sloped] & u^1 \arrow[equal,d, "u^{\mathrm{pm}}"] \\
t^2                                                                                        & u^2                             
\end{tikzcd}
\end{tabular}
\end{center}

A first-order term $t$ is interpreted as a proof of
(parametric) equality between two copies of itself, and an equality proof
$M: t=u$ is translated into a proof of 
(parametric) equality between a copy of
$t$ and a copy of $u$. While the choice to translate $M$ as an
equality between $t^1$ and $u^2$ is ad hoc (in the sense that
it is not imposed by the translation), it is notable that
equality proofs are translated into (parametric) equality proofs
without the need of higher-order structures (that do not 
exist in this framework).

\nocite{*}
\bibliographystyle{eptcs}
\bibliography{generic}
\end{document}

%% file: macro.tex
\newcommand\refle{\mathrm{refl}}
\newcommand\peel{\mathrm{peel}}
\newcommand\letc{\mathrm{let}}
\newcommand\inc{\mathrm{in}}
\newcommand\nullt{\mathrm{null}}
\newcommand\efq{\mathrm{efq}}

\newcommand\HRO{\ensuremath{\mathbf{HRO}}}
\newcommand\HROo{\ensuremath{\mathbf{HRO^{o}}}}
\newcommand\M{\ensuremath{\mathbf{M}}}
\newcommand\Pm{\ensuremath{\mathrm{pm}}}
\newcommand\HAw{\ensuremath{\mathbf{HA^{\omega}}}}
\newcommand\NHAw{\ensuremath{\mathbf{N\mbox{-}HA^{\omega}}}}
\newcommand\EHAw{\ensuremath{\mathbf{E\mbox{-}HA^{\omega}}}}
\newcommand\IHAw{\ensuremath{\mathbf{I\mbox{-}HA^{\omega}}}}

\newcommand\LEHAw{\ensuremath{\mathbf{\lambda E\mbox{-}HA^{\omega}}}}
\newcommand\LHAw{\ensuremath{\mathbf{\lambda HA^{\omega}}}}

\newcommand\Ext{\ensuremath{=^{\mathrm{ext}}}}
\newcommand\extpm{\ensuremath{\mathrm{ext}}}
\newcommand\extp{\mathrm{e}}
\newcommand\wfp{\mathbf{wf}}
\newcommand\Eqpm{\ensuremath{=^{\mathrm{pm}}}}

\newcommand\Equiv{\ensuremath{\mathrm{Equiv}}}
\newcommand\T{{\textrm{T}}}
\newcommand\Nat{\mathbf{N}}
\newcommand\succe{\mathrm{S}}

\newcommand\Sym{\mathbf{Sym}}
\newcommand\sympm{\mathrm{sym}^{\mathrm{pm}}}
\newcommand\reflpm{\mathrm{refl}^{\mathrm{pm}}}
\newcommand\Trans{\mathbf{Trans}}
\newcommand\transpm{\mathrm{trans}^{\mathrm{pm}}}
\newcommand\Elim{\mathrm{Elim}}
\newcommand\apppm{\mathrm{app}}
\newcommand\Rec{\mathrm{Rec}}
\newcommand\Hquote{\mathrm{quote}}
\newcommand\trans{\mathrm{trans}}
\newcommand\FV{\mathbf{FV}}
\newcommand\Dom{\mathbf{Dom}}
\newcommand\Ind{\mathrm{Ind}}
\newcommand\Collaps{\mathrm{Collaps}}

\newtheorem{thm}{Theorem}
\newtheorem{lemma}{Lemma}
\newtheorem{conj}{Conjecture}
\newtheorem{fact}{Fact}
\newtheorem{cor}{Corollary}
\newtheorem{prop}{Proposition}